\documentclass[a4paper,UKenglish,cleveref, autoref, thm-restate]{lipics-v2021}
\pdfoutput=1 %uncomment to ensure pdflatex processing (mandatatory e.g. to submit to arXiv)
\hideLIPIcs  %uncomment to remove references to LIPIcs series (logo, DOI, ...), e.g. when preparing a pre-final version to be uploaded to arXiv or another public repository

\usepackage[T1]{fontenc}
\usepackage{enumitem}
\usepackage{mathdots}
\usepackage[normalem]{ulem}
\usepackage{xspace}
\xspaceaddexceptions{]\}}
\usepackage{comment}
\usepackage{tabu}
\usepackage{framed}
\usepackage{float,wrapfig}

\newcommand{\eps}{\varepsilon}
\renewcommand{\Pr}{\operatorname*{\mathbf{Pr}}}
\newcommand{\Ex}{\operatorname*{\mathbf{E}}}
\newcommand{\pigeon}{Pigeonhole Equal Sums\xspace}

\newcommand{\poly}{\operatorname{\mathrm{poly}}}
\newcommand{\polylog}{\poly\log}

\newcommand{\N}{\mathbb{N}}

\newcommand{\Z}{\mathbb{Z}}

\EventEditors{John Q. Open and Joan R. Access}
\EventNoEds{2}
\EventLongTitle{42nd Conference on Very Important Topics (CVIT 2016)}
\EventShortTitle{CVIT 2016}
\EventAcronym{CVIT}
\EventYear{2016}
\EventDate{December 24--27, 2016}
\EventLocation{Little Whinging, United Kingdom}
\EventLogo{}
\SeriesVolume{42}
\ArticleNo{23}
\title{A Faster Algorithm for Pigeonhole Equal Sums}

\author{Ce Jin}{MIT}{cejin@mit.edu}{}{Supported by NSF grants CCF-2129139 and CCF-2127597. Work done while visiting the Simons Institute for the Theory of Computing.}

\author{Hongxun Wu}{UC Berkeley}{wuhx@berkeley.edu}{}{Supported by Avishay Tal's Sloan Research Fellowship, NSF CAREER Award CCF-2145474, and Jelani Nelson's ONR grant N00014-18-1-2562.}

\authorrunning{C. Jin and H. Wu} 

\Copyright{Ce Jin and Hongxun Wu}

   \ccsdesc[500]{Theory of computation~Design and analysis of algorithms}

\keywords{Subset Sum, Pigeonhole, PPP} 

\acknowledgements{We thank Ryan Williams for useful discussions.}%optional

\nolinenumbers %uncomment to disable line numbering

\begin{document}

 \maketitle
 %\thispagestyle{empty}
 %\newpage

\begin{abstract}
An important area of research in exact algorithms is to solve Subset-Sum-type problems faster than meet-in-middle.   
In this paper we study \emph{Pigeonhole Equal Sums}, a total search problem proposed by Papadimitriou (1994): given $n$ positive integers $w_1,\dots,w_n$ of total sum $\sum_{i=1}^n w_i < 2^n-1$, the task is to find two distinct subsets $A, B \subseteq [n]$ such that $\sum_{i\in A}w_i=\sum_{i\in B}w_i$.

Similar to the status of the Subset Sum problem, the best known algorithm for Pigeonhole Equal Sums  runs in $O^*(2^{n/2})$ time, via either meet-in-middle or dynamic programming (Allcock, Hamoudi, Joux, Klingelh\"{o}fer, and Santha, 2022).

Our main result is an improved algorithm for Pigeonhole Equal Sums in $O^*(2^{0.4n})$ time. We also give a polynomial-space algorithm in $O^*(2^{0.75n})$ time. 
Unlike many previous works in this area, our approach does not use the representation method, but rather exploits a simple structural characterization of input instances with few solutions. 
\end{abstract}

\section{Introduction}

The Subset Sum problem is an important NP-hard problem in computer science: given positive integers $w_1,w_2,\dots,w_n$ and a target integer $t$, find a subset $A\subseteq [n]$ such that $\sum_{i\in A}w_i=t$. Subset Sum can be solved in $O(2^{n/2})$ time by a simple meet-in-middle algorithm \cite{horowitz1974computing}, and an important open problem is to improve it to $O(2^{(1/2-\eps)n})$.
A long line of research attempts to solve Subset Sum faster using the representation method \cite{hgj10} and connections to uniquely decodable code pairs \cite{AKKN16,AustrinKKN18,Tilborg78}, but these techniques have so far only succeeded on average-case inputs \cite{hgj10,BeckerCJ11,asiacrypt} or restricted classes of inputs \cite{AustrinKKN15,AKKN16}.
Nevertheless, significant progress has been made for other variants of Subset Sum, including 
 Equal Sums \cite{esaequal19}, 2-Subset Sum and Shifted Sums \cite{esa22} and more general subset balancing problems \cite{soda22}, as well as Subset Sum in other computational settings such as Merlin--Arthur protocols \cite{Nederlof17}, low-space algorithms \cite{BansalGN018,NederlofW21}, quantum algorithms \cite{esa22}, and algorithms with lower-order run time improvements \cite{randomapprox}.
The general hope is that the tools developed for solving these variant problems might one day help solve the original Subset Sum problem. 

In this paper we study an interesting variant of Subset Sum called \emph{Pigeonhole Equal Sums}:
\newcommand{\defproblem}[3]{
  \vspace{2mm}
  \vspace{1mm}
\noindent\fbox{
  \begin{minipage}{0.95\textwidth}
  #1 \\
  {\bf{Input:}} #2  \\
  {\bf{Output:}} #3
  \end{minipage}
  }
  \vspace{2mm}
}

\defproblem{\textsc{Pigeonhole Equal Sums} \cite{Papadimitriou94}}
{positive integers $w_1,w_2,\dots,w_n$, with promise $\sum_{i=1}^n w_i < 2^n - 1$.
}
{two different subsets $A,B\subseteq [n]$ such that $\sum_{i\in A} w_i= \sum_{i\in B} w_i$.
}
\\
Since there are $2^n$ subsets $S\subseteq [n]$ with only $2^{n}-1$ possible subset sums $\sum_{i\in S}w_i \in \{0,1,\dots,2^{n}-2\}$ due to the promise, the pigeonhole principle guarantees that there exists a pair of subsets with the same subset sum.

Pigeonhole Equal Sums was introduced by Papadimitriou \cite{Papadimitriou94} as a natural example problem in the total search complexity class PPP.  
This problem has received attention in the TFNP literature \cite{BanJPPR19,SotirakiZZ18}, and is conjectured to be PPP-complete~\cite{Papadimitriou94}. 

From the algorithmic point of view, the current status of Pigeonhole Equal Sums is quite similar to that of the Subset Sum problem: a simple binary search with meet-in-middle solves Pigeonhole Equal Sums in $O^*(2^{n/2})$ time (see \cref{sec:prelim}).\footnote{We use $O^*(\cdot)$ to hide $\poly(n)$ factors.}
Allcock, Hamoudi, Joux, Klingelh\"{o}fer, and Santha~\cite[Theorem 6.2]{esa22} gave another $O^*(2^{n/2})$-time algorithm based on dynamic programming (which is analogous to the alternative $O^*(2^{n/2})$-time Subset Sum algorithm  from \cite{AKKN16}\footnote{See also \url{https://youtu.be/cHimhXXIwcg?t=454}.}). 
It remains open whether $O(2^{(1/2-\eps)n})$ time is possible for Pigeonhole Equal Sums.
Improvement of such type was achieved for the Equal Sums problem (without the pigeonhole promise) by Mucha, Nederlof, Pawlewicz, and W\k{e}grzycki \cite{esaequal19} via the representation method with $O(3^{(1/2-\eps)n})$ run time for some $\eps>0.01$, but this result has no direct implications for Pigeonhole Equal Sums (for which the known $O^*(2^{n/2})$ time bound is already much better than $O(3^{n/2})$).

\subsection{Our results}
We give an algorithm that solves Pigeonhole Equal Sums faster than the previous $O^*(2^{n/2})$ running time \cite{esa22}.
\begin{theorem}[Main]
   \label{thm:main}
Pigeonhole Equal Sums can be solved by a randomized algorithm in $O^*(2^{0.4n})$  time.
\end{theorem}
Surprisingly, unlike previous works on other variants of Subset Sum, our algorithm does not use the representation method \cite{hgj10} or tools from coding theory \cite{AKKN16,AustrinKKN18,Tilborg78}. Instead, our main insight is a simple structural characterization of Pigeonhole Equal Sums instances with few solutions.

 Our techniques also yield a fast polynomial-space algorithm for Pigeonhole Equal Sums, in an analogous way to the previous $O(3^{(1-\eps)n})$-time polynomial-space algorithm for Equal Sums \cite{esaequal19}.
\begin{theorem}
   \label{thm:mainlowspace}
Pigeonhole Equal Sums can be solved by a randomized algorithm in $O^*(2^{0.75n})$ time and $\poly(n)$ space. 
\end{theorem}
For comparison, a straightforward algorithm based on binary search solves Pigeonhole Equal Sums in $\poly(n)$ space and $O^*(2^n)$ time (see the beginning of \cref{sec:poly}).

\cref{thm:main} and \cref{thm:mainlowspace} will be proved in \cref{sec:mainalgo} and \cref{sec:poly} respectively.

\section{Preliminaries}
\label{sec:prelim}

Denote $[n]=\{1,\dots,n\}$.
Let $O^*(\cdot), \Omega^*(\cdot)$ hide $\poly(n)$ factors, where $n$ is the number of input integers in the \pigeon problem.

 Denote $w(A) = \sum_{i\in A} w_i$ for $A\subseteq [n]$. 
The pigeonhole promise states $w([n])< 2^n -1$.

For a predicate $p$ we define $\mathbf{1}[p] =1$ if $p$ is true and $\mathbf{1}[p] =0$ if $p$ is false.

 We need the following well-known lemma.
\begin{lemma}[Counting subset sums via meet-in-middle \cite{horowitz1974computing}]
   \label{lem:mim}
   Given integers $w_1,\dots,w_n$ and $t$, we can compute $\#\{S\subseteq [n]:w(S) \le t\}$ in $O^*(2^{n/2})$ time. Moreover, we can list $S\subseteq [n]$ such that $w(S) \le t$ in $O^*(1)$ additional time per $S$.
\end{lemma}
\begin{proof}
   Divide $[n]$ into $S_1=\{1,\dots,\lfloor n/2\rfloor \}$ and $S_2=[n]\setminus S_1$, and every subset $S\subseteq [n]$ can be represented as $X\uplus Y, X\subseteq S_1,Y\subseteq S_2$. Compute and sort the two lists $A = \{w(X)\}_{X\subseteq S_1}$ and $B = \{w(Y)\}_{Y\subseteq S_2}$ of length $O(2^{n/2})$ each. Then for each $w(X)\in A$ we accumulate $| B\cap (-\infty, t-w(X)]|$ to the answer. It is easy to augment this algorithm to support listing.
\end{proof}

\paragraph*{Pigeonhole Equal Sums via binary search}
The following simple binary-search algorithm (described in \cite[Remark 6.9 of arXiv version]{esa22} and attributed to an anonymous referee) solves \pigeon in $O^*(2^{n/2})$ time: Maintain an interval $\{\ell,\ell+1,\dots,r\}$ (initialized to $\ell=0,r=2^{n}-2$)  that satisfies the pigeonhole invariant $r-\ell+1 < \#\{S\subseteq [n]: \ell\le w(S)\le r\}$.
Initially this invariant is satisfied due to $w([n])\le 2^n -2$.
While $r>\ell$, pick the middle point $m = \lfloor \tfrac{\ell + r}{2} \rfloor$, and use meet-in-middle (\cref{lem:mim}) to compute $c_1 = \#\{S\subseteq [n]: \ell\le w(S)\le m\}$ and $c_2 = \#\{S\subseteq [n]: m+1\le w(S)\le r\}$ in $O^*(2^{n/2})$ time. Then we shrink the interval to $\{\ell,\dots,m\}$ if $m-\ell+1< c_1$, or to $\{m+1,\dots,r\}$ if $r-m< c_2$ (the invariant guarantees that at least one holds). After $\lceil \log_2(2^n-1)\rceil = n$ iterations we shrink to a singleton interval $\ell=r$. By the invariant, there exist two different $S_1,S_2\subseteq [n]$ such that $w(S_1)=w(S_2)=\ell$, and we can report such $S_1,S_2$ using meet-in-middle (\cref{lem:mim}).

This binary-search strategy will be used in our improved algorithms as well.

\section{The improved algorithm}
\label{sec:mainalgo}
Let the $n$ input integers be sorted as $0<w_1< w_2< \dots< w_n$ (assuming no trivial solution $w_i= w_j$ exists).
\paragraph*{An assumption on prefix sums}
If any proper prefix $\{w_1,\dots,w_i\}$ ($i\le n-1$) already satisfies the pigeonhole promise $w([i])< 2^i-1$, then we can instead 
solve the smaller \pigeon instance $\{w_1,\dots,w_i\}$ and obtain $A,B\subseteq [i], A\neq B$ with $w(A) =w(B)$. Hence, without loss of generality we assume such prefix does not exist, i.e.,
\begin{equation}
   \label{eqn:lb}
   w([i])\ge 2^i-1
 \text{ for all } i\in [n-1].
\end{equation}

\paragraph*{Frequencies $f_t$ and parameter $d$}
The \emph{frequency} (also called bin size) of $t\in \N$ is the number of input subsets achieving sum~$t$, denoted as
$   f_t = \# \{ S\subseteq [n]: w(S) = t\}$. 
Since $w([n])<2^n-1$, we know $f_t = 0$ for all $t\ge 2^n-1$, and
\begin{equation}
   \label{eqn:allsol}
   \sum_{0\le t<2^n}f_t = 2^n.
\end{equation}

Two different subsets achieving equal subset sum $t$ imply $f_t> 1$. This motivates the following parameter,
\begin{equation}
   \label{eqn:defndel}
 d = \sum_{0\le t< 2^n} \max\{0,f_t-1\}, 
\end{equation}
which counts the (non-trivial) equality relations among  all the $2^n$ subset sums.
Using \cref{eqn:allsol}, we can rewrite \cref{eqn:defndel} as $d=\sum_{0\le t<2^n}(f_t - \mathbf{1}[f_t\ge 1]) = 2^n - \sum_{0\le t<2^n}\mathbf{1}[f_t\ge 1]$, and thus obtain
\begin{equation}
   \label{eqn:zeros}
 d = \#\{0\le t<2^n:  f_t=0 \},
\end{equation}
which counts the non-subset-sums in $\{0,1,\dots,2^{n}-1\}$. In particular, $d<2^n$.

The equivalence between \cref{eqn:defndel} and \cref{eqn:zeros} is powerful.
In the following we will give two different algorithms for \pigeon. The first one works for small $d$ by analyzing the structure of input instances with few non-subset-sums (by \cref{eqn:zeros}). The second one works when $d$ is large and hence there are many solutions (by \cref{eqn:defndel}) which allow a subsampling approach. These two algorithms are summarized as follows:
\begin{lemma}
   \label{thm:smalldelta}
   Given parameter $\Delta \le  2^n/(3n^2)$,  \pigeon with $d\le \Delta$ can be solved deterministically 
  in $O^*(\sqrt{\Delta})$ time.
\end{lemma}

\begin{lemma}
   \label{thm:larged}
Given parameter $2^{n/2} \le \Delta < 2^n$, \pigeon with $d\ge \Delta$ can be solved in $O^*((2^{2n}/\Delta)^{1/3})$ time by a randomized algorithm.
\end{lemma}

Combining these two lemmas implies our main result:
\begin{proof}[Proof of \cref{thm:main}]
Set $\Delta = 2^{0.8n}$  so that the two time bounds in \cref{thm:smalldelta} and \cref{thm:larged} are balanced to $O^*(2^{0.4n})$.
Given an instance of \pigeon (with unknown $d$), we run both algorithms in parallel, and return the answer of whichever terminates first.
\end{proof}

\subsection{Small \texorpdfstring{$d$}{d} case via structural characterization}
\label{sec:smalld}
In this section we prove \cref{thm:smalldelta}. Assume $d\le \Delta\le 2^n/(3n^2)$ and $\Delta$ is known.

Since $f_t = 0$ for all $w([n]) <t<2^{n}$, from \cref{eqn:zeros} we know $d\ge 2^n -1 - w([n])$, and hence $w([n]) \ge 2^n-1-d \ge 2^n-1-\Delta$. Combined with \cref{eqn:lb} for $i\in [n-1]$, we get the following lower bound
  \begin{equation}
     \label{eqn:lbn}
   w([i])\ge 2^{i} -1- \Delta \text{ for all } i\in [n].
  \end{equation}

The key step is to complement \cref{eqn:lbn} with a nearly matching upper bound:
\begin{lemma}
  For all $i\in [n]$, 
  \begin{equation}
     \label{eqn:tempub}
  w_i\le 2^{i-1}  + \Delta.
  \end{equation}
\end{lemma}
Summing \cref{eqn:tempub} over $i$ gives
\begin{equation}
   \label{eqn:ub}
 w([i])\le 2^{i} - 1 + i\Delta
\end{equation}
for all $i\in [n]$.
\begin{proof}
Fix $i\in [n]$.   Let $M$
    be the number of subsets $S\subseteq [n]$ with $w(S)<w_i$.
 Since $w_i<w_{i+1}<\dots<w_n$, any such $S$ must be contained in $[i-1]$, and thus $M\le 2^{i-1}$. On the other hand, 
$M= \sum_{t=0}^{w_i-1}f_t \ge   w_i - \#\{0\le t<w_i: f_t=0\} \ge w_i - d$ by \cref{eqn:zeros}. Hence, $w_i\le M+d \le 2^{i-1}+\Delta$. \end{proof}

Comparing  \cref{eqn:lbn} with \cref{eqn:ub} gives the lower bound
   \begin{align*}
      w_i = w([i]) - w([i-1])
 \ge (2^i-1-\Delta) - (2^{i-1} - 1 +(i-1)\Delta) 
 = 2^{i-1}-i\Delta,
   \end{align*}
   which is very close to the upper bound from \cref{eqn:tempub}. Together we get
  \begin{equation}
     \label{eqn:precise}
 w_i - 2^{i-1} \in [-i\Delta, \Delta]    \end{equation}
for all $i\in[n]$.

\cref{eqn:precise} gives a very rigid structure of the large input numbers. In the next lemma we exploit this structure to improve the naive meet-in-middle subset sum counting algorithm  from \cref{lem:mim}.

\begin{lemma}
   \label{lem:query}
For any given $T<2^n$, we can compute $\sum_{t=0}^T f_t$ in $O^*(\sqrt{\Delta})$ time.
\end{lemma}
\begin{proof}
Let $i^*$ be the minimum $i^*\in [n]$ such that 
$ 2^{i^*}\ge 3n^2\Delta$,
which exists by our assumption $\Delta\le 2^n/(3n^2)$.
Let $A =\{1,2,\dots,i^*\}$ and $B = \{i^*+1,\dots,n\}$.

By \cref{eqn:ub},
$w(A) < 2^{i^*}  + n\Delta$.

For every $B'\subseteq B$, by \cref{eqn:precise} we have
\begin{align*}
   \big \lvert w(B') - \sum_{j\in B'}2^{j-1}\big \rvert \le \sum_{j\in B'}|w_j - 2^{j-1}|
    \le \sum_{j\in B'} j\Delta 
    \le n^2\Delta.
\end{align*}
In other words, the subset sums of $\{w_j\}_{j\in B}$ are $n^2\Delta$-additively approximated by the subset sums of $\{2^{j-1}\}_{j\in B}$. The subset sums of the latter set  form an arithmetic progression $\{k\cdot 2^{i^*} : 0\le k< 2^{n-i^*}\}$, namely all $n$-bit binary numbers whose lowest $i^*$ bits are zeros. Notably, this arithmetic progression is very sparse: its difference $2^{i^*}$ is large enough compared to $w(A)<2^{i^*}+n\Delta$.

Given query $T$, we want to count the number of pairs $(A',B')$ $(A'\subseteq A, B'\subseteq B)$ such that $w(A')+w(B')\le T$.
To do this, we enumerate $B'\subseteq B$, and consider three cases (the non-trivial case is Case 3, where $w(B')$ and $\sum_{j\in B'}2^{j-1}$ are close to $T$):
\begin{itemize}
   \item \textbf{Case 1:} $\sum_{j\in B'}2^{j-1}\le T - 2^{i^*}-(n+n^2)\Delta$.

   Then,  for all $A'\subseteq A$, we have $w(A')+w(B') \le w(A)+w(B') \le (2^{i^*}+n\Delta) + (n^2\Delta + \sum_{j\in B'}2^{j-1}) \le  T$. Hence $B'$ contributes $2^{|A|}$ many pairs $(A',B')$.

   \item \textbf{Case 2:} $\sum_{j\in B'}2^{j-1}> T + n^2\Delta$.

      Then, for all $A'\subseteq A$, we have $w(A')+w( B')\ge w(B') \ge \sum_{j\in B'}2^{j-1} - n^2\Delta >T$.
Hence $B'$ does not contribute any pairs $(A',B')$. 
\item \textbf{Case 3:} otherwise, $\sum_{j\in B'}2^{j-1} \in (T - 2^{i^*}-(n+n^2)\Delta, T + n^2\Delta]$.

This interval has length $2^{i^*} + (n+n^2)\Delta+n^2\Delta \le 2\cdot 2^{i^*}$ by our choice of $i^*$. Since $\sum_{j\in B'}2^{j-1}$ is a multiple of $2^{i^*}$ in this interval, it has at most two possibilities, namely $2^{i^*}\cdot \lfloor \frac{T -(n+n^2)\Delta}{2^{i^*}}\rfloor$ and $2^{i^*}\cdot\left (\lfloor \frac{T -(n+n^2)\Delta}{2^{i^*}}\rfloor+1\right )$, and then $B'$ is uniquely determined by the binary decomposition of $\sum_{j\in B'}2^{j-1}$. For each possible $B'$, we count the number of $A'\subseteq A$ such that $w(A') \le T - w(B')$ using meet-in-middle (\cref{lem:mim}) with time complexity $O^*(2^{|A|/2}) = O^*(2^{i^*/2}) = O^*(\sqrt{\Delta})$ by the definition of $i^*$.
\end{itemize}

Note that in $O^*(1)$ time we can easily find the (at most two) subsets $B'$ satisfying Case~3, and also count the total contribution of Case~1. Hence the overall time complexity is $O^*(\sqrt{\Delta})$.
\end{proof}
Using \cref{lem:query} we can solve \pigeon using binary search, in the same way as described in the last paragraph of \cref{sec:prelim}. The running time is $O^*(\sqrt{\Delta})$. This finishes the proof of \cref{thm:smalldelta}.

\subsection{Large \texorpdfstring{$d$}{d} case via subsampling}
In this section we prove \cref{thm:larged}. 
Assume $2^{n/2} \le \Delta\le d<2^n$, and $\Delta$ is known. 
We first use
 $d = \sum_{0\le t< 2^n} \max\{0,f_t-1\}$ (\cref{eqn:defndel}) to show that many subset sums $t$ have large $f_t$, which then allows us to use subsampling to speed up the modular dynamic programming approach of \cite{esa22,AKKN16}.
\begin{lemma}
   \label{lem:existsj}
There exists  a $ j \in \{0,1,\dots, n-1\}$ such that 
$\# \{t : f_t > 2^{j} \} > \frac{\Delta}{2^{j+1}n}$.
\end{lemma}
\begin{proof}
By definition of $d$ in \cref{eqn:defndel},
\begin{equation}
   \label{eqn:temptemp}
 \Delta\le d = \sum_{t: f_t>1}(f_t-1) \le  \sum_{0\le j< n} \# \{t : 2^{j} < f_t \le 2^{j+1} \}\cdot (2^{j+1}-1) .
\end{equation}
 If the claimed inequality fails for all $j$, then
 \[ [\text{RHS of \cref{eqn:temptemp}}] 
   \le \sum_{0\le j< n} \frac{\Delta}{2^{j+1}n} \cdot (2^{j+1}-1) <  \Delta,\]
a contradiction.
\end{proof}

Our algorithm enumerates all 
$ j \in \{0,1,\dots, n-1\}$ (increasing the time complexity by a factor of $n = O^*(1)$), and from now on we assume $j$ satisfies the inequality in  \cref{lem:existsj}.
Define
\begin{equation}
  h := 2^j+1\ge 2,\;\; m:= \left \lceil \frac{\Delta}{2^{j+1}n} \right \rceil > \frac{\Delta}{2hn} ,\; \text{ and } X := \{t\in [2^n] : f_t \ge h\}. \label{eqn:defnhmx}
\end{equation}
Here we defined the set $X$ of frequent subset sums only for the sake of analysis. By \cref{lem:existsj},~
\begin{equation}
\label{eqn:xgem}
 |X|\ge m.   
\end{equation}
 Readers are encouraged to focus on the case of $h=2$ and $m\ge \Omega^*(\Delta)$ (which is the hardest case for our algorithm) at first read.

We first describe the behavior of our algorithm:  Let $p\in [P,2P]$ be a uniformly random prime (for some parameter $2\le P\le 2m$ to be determined later in the ``Time complexity'' paragraph).
For each $r\in \Z_p$, define bin $B_r:= \{S\subseteq [n]: w(S)\equiv r\pmod{p}\}$. 
The algorithm picks a random bin index $r^*\in \Z_p$, and subsamples $C\subseteq B_{r^*}$ by keeping each $S\in B_{r^*}$ with probability $\alpha$ independently (for some $0<\alpha\le \frac{1}{2h}$ to be determined later in the ``Success probability'' paragraph).
Finally, a pair of distinct $S,S'\in C$ with $w(S)=w(S')$ is reported (if exists).

Now we explain how to implement the algorithm above via dynamic programming (DP) similarly to \cite{esa22,AKKN16}. Build the DP table $D_{i,r} = \#\{S \subseteq  [i] : w(S) \equiv r \pmod{p}\}$ (where $0\le i \le n$ and $r\in \Z_p$)  in $O^*(p)$ overall time via the transition $D_{i,r} = D_{i-1,r}+D_{i-1, (r-w_i)\bmod p}$ with initial values $D_{0,r} = \textbf{1}[r = 0]$. This DP computes the size of every bin $|B_r| = D_{n,r}$.  
Furthermore, for any bin $B_r$ and integer $k\in [|B_r|]$, we can report the rank-$k$ set $S$ in $B_r$  (in lexicographical order, where larger indices are compared first) by backtracing in the DP table in $O^*(1)$ time. Then,
in order to subsample a collection of sets $C\subseteq B_{r^*}$ at rate $\alpha$, we can first subsample their ranks in $[|B_{r^*}|]$ (in near-linear time in the output size, see e.g., \cite{BringmannP17sample}), and then recover the actual sets by backtracing.

\paragraph*{Success probability}
We study how the frequent subset sums, $X=\{t:f_t\ge h\}$, are distributed to the bins modulo a random prime $p$, using an argument similar to 
\cite{AKKN16}.
Setting
\begin{equation}
\label{eqn:defnk}
  k:=\big \lceil \frac{m}{4P} \big\rceil, 
\end{equation}
  the following lemma shows that the bin $B_{r^*}$ receives at least $k$ frequent subset sums, with $\Omega^*(1)$ probability.
\begin{lemma}
\label{lem:frequentsums}
    With at least $\Omega(1/n)$ probability over the choice of prime $p\in [P,2P]$ and $r^*\in \Z_p$, 
    there are at least $k$ integers $t\in \N$ such that $\#\{S\in B_{r^*} : w(S) = t\} \ge h$. 
\end{lemma}
\begin{proof}
Since $|X|\ge m$ by \cref{eqn:xgem}, we arbitrarily pick $X'\subseteq X$ with $|X'|=m$ for the sake of analysis.
Let $c_{r,p}:= \{t\in X': t \equiv r\pmod{p}\}$. 
   Then, 
   \begin{align*}
\Ex_{p\in [P,2P]} \big [\sum_{r\in \Z_p} c_{r,p}^2\big ] &=  \sum_{x\in X',y\in X'}\Pr_{p\in [P,2P]}[p\mid x-y] \\
& \le m + m^2\cdot \frac{\log_P 2^{n}}{\Omega(P/\ln P)} \tag{by $|x-y|\le 2^n$ and the density of primes}
\\ 
& \le O(n\cdot m^2/P).  \tag{by the assumption that $P\le 2m$}
   \end{align*}
Then by Markov's inequality, with $0.9$ success probability over the choice of $p$, we have $\sum_{r\in \Z_p} c_{r,p}^2 \le O(n\cdot m^2/P)$.
Conditioned on this happening, by Cauchy--Schwarz inequality we have 
\begin{align*}
\sum_{r\in \Z_p} \mathbf{1}[c_{r,p}\ge \tfrac{m}{2p}] & \ge \frac{\left (\sum_{r\in \Z_p}\mathbf{1}[c_{r,p}\ge \tfrac{m}{2p}]\cdot c_{r,p}\right )^2}{\sum_{r\in\Z_p}c_{r,p}^2} \\
& \ge  \frac{\big ((\sum_{r\in \Z_p}c_{r,p}) - p\cdot \tfrac{m}{2p}\big )^2}{O(n\cdot m^2/P)}=\frac{(|X'|-m/2)^2}{O(n\cdot m^2/P)} =\frac{(m/2)^2}{O(n\cdot m^2/P)}= \Omega(P/n),
\end{align*}
and hence, by our choice of  $k=\big \lceil \frac{m}{4P} \big\rceil \le \big \lceil \frac{m}{2p} \big\rceil$,
\[\Pr_{r^*\in \Z_p} [c_{r^*,p}\ge  k]\,\ge\, \Pr_{r^*\in \Z_p} [c_{r^*,p}\ge  \tfrac{m}{2p}] \,\ge\, \frac{\Omega(P/n)}{p} = \Omega(1/n).\]
  Conditioned on $c_{r^*,p}\ge  k$ happening,  we have at least $k$ integers $t\in X'\subseteq X$ such that $t\equiv r^* \pmod{p}$. By definitions of $B_{r^*}$ and $X$, this implies that there are at least $k$ integers $t\in \N$ such that $\#\{S\in B_{r^*} : w(S) = t\} \ge h$, with overall success probability at least $0.9\cdot \Omega(1/n) = \Omega(1/n)$ over the choice of $p$ and $r^*$.
\end{proof}

  Recall our algorithm subsamples $C\subseteq B_{r^*}$ at rate $\alpha\in (0,\frac{1}{2h}]$, and fails iff $w(S)$ are distinct for all $S\in C$. The failure probability of this step can be derived from the following lemma:
  \begin{lemma}
     \label{lem:balls}
     Let $B'$ be a collection of $kh$ colored balls $(h\ge 2, k\ge 1)$, with exactly $h$ balls of color $i$ for each  color $i\in [k]$.
     Let $C' \subseteq B'$ be an i.i.d.\ subsample at rate $\alpha \in [0, \frac{1}{2h}]$. Then $C'$ contains distinct colors with at most $\exp(-kh(h-1)\alpha^2/4)$ probability.
  \end{lemma}
  \begin{proof}
  For each color $i\in [k]$, by Bernoulli's inequality,
 the probability that $C'$ includes exactly two balls of color~$i$ is $\binom{h}{2}\alpha^2(1-\alpha)^{h-2} \ge \binom{h}{2}\alpha^2\big (1-(h-2)\alpha\big ) \ge \binom{h}{2}\alpha^2/2$.
Hence, the probability that $C'$ includes at most one ball of every  color $i\in [k]$ is at most 
   $ \big ( 1 - \binom{h}{2}\alpha^2/2\big )^{k} \le \exp\big (-k\binom{h}{2}\alpha^2/2\big ) =  \exp(-kh(h-1)\alpha^2/4)$.
  \end{proof}
 We think of each set $S\in B_{r^*}$ as a ball of color $w(S)$,
 and apply \cref{lem:balls} to the 
  $k$ integers (colors) $t\in \N$ ensured by \cref{lem:frequentsums},  each having at least $h$ sets (balls) $S\in B_{r^*}$ with $w(S) = t$. 
  We set the sample rate to be
  \begin{equation}
  \label{eqn:defnalpha}
  \alpha := \frac{1}{2h\sqrt{k}} \le \frac{1}{2h}.
  \end{equation}
  Then the failure probability of the subsampling step is at most
 \[ \exp (-kh(h-1)\alpha^2/4 )  = \exp(-\tfrac{h-1}{16h}) \le \exp(-1/32), \] 
   %(where we used $h\ge 2$).

  Overall, the probability that the algorithm successfully finds a solution is at least $\Omega(1/n) \cdot (1-\exp(-1/32)) \ge \Omega(n^{-1})$.

\paragraph*{Time complexity}
The mod-$p$ DP runs in $O^*(p)\le O^*(P)$ time.
Since the bins have total size $\sum_{r\in \Z_p}|B_r|=2^n$, the chosen bin $B_{r^*}$ has expected size $\Ex_{r^*\in \Z_p}[|B_{r^*}|] = 2^n/p \le 2^n/P$, and hence the subsample $C\subseteq B_{r^*}$ has expected size $\Ex[|C|] \le \alpha 2^n/P$. To detect a solution $S,S'\in C$ with $w(S)=w(S')$, we simply sort $C$  in near-linear time.  Hence the total expected running time is $O^*(P + \alpha 2^n/P)$. By Markov's inequality,
with probability at least $1-n^{-10}$, the algorithm terminates in $O^*(P + \alpha 2^n/P)$ time. By a union bound, the algorithm successfully finds a solution in time $O^*(P + \alpha 2^n/P)$ with probability at least $\Omega(n^{-1})-n^{-10}\ge \Omega(n^{-1})$. This success probability can be boosted to $0.99$ by repeating the algorithm $O(n)$ times.

Recall from \cref{eqn:defnk,eqn:defnalpha} that 
$ \alpha = \frac{1}{2h\sqrt{k}}  =\frac{1}{2h \sqrt{\lceil m/4P\rceil }} \le  \frac{\sqrt{P}}{h\sqrt{m}}$, 
so the run time is (ignoring $\poly(n)$ factors)
\[ P + \alpha 2^n/P \le  P + \frac{2^n }{h\sqrt{mP}}.\]

Recall $h=2^j+1$ (where $0\le j\le n-1$) and $m= \lceil \frac{\Delta}{2^{j+1}n} \rceil $, and hence $hm < h(1+  \frac{\Delta}{2^{j+1}n}) \le h+ \frac{\Delta}{n} < (2^{n-1}+1) + \frac{2^n}{n}\le 2^n$ (assuming $n\ge 3$).
Now we set 
\[P := 2m\cdot \min\left \{1, \big (\frac{2^n}{hm^2}\big )^{2/3}\right \},\] and we first need to verify the requirement $2\le P \le 2m$ introduced earlier: The upper bound is obvious. To see the lower bound, note that $2m\ge 2$ and 
$ 2m\cdot \left ( \frac{2^n}{hm^2}\right )^{2/3} =   2  \left (\frac{2^{2n}}{h^2m}\right )^{1/3} \ge 2  \left (\frac{2^{2n}}{(hm)^2}\right )^{1/3} \ge 2$ (using the inequality $hm\le 2^n$ we just showed).

Hence, the overall running time is at most (ignoring $\poly(n)$ factors)
\begin{align*}
   P + \frac{2^n }{h\sqrt{mP}}  &\le 2m\left (\frac{2^n}{hm^2}\right )^{2/3} + \frac{2^n}{h\sqrt{m\cdot 2m}}\cdot \max\left \{1 ,\left (\frac{hm^2}{2^n}\right )^{1/3} \right \}\\
   &  = 2\cdot \frac{2^{2n/3}}{h^{2/3}m^{1/3}}  +\frac{1}{\sqrt{2}}\max \left \{ \frac{2^n}{hm} , \frac{2^{2n/3}} {h^{2/3} m^{1/3}}\right \}\\
   & \le O\left (\frac{2^{2n/3}}{h^{2/3}m^{1/3}} + \frac{2^n}{hm}\right )\\
   & \le O^*\left (\frac{2^{2n/3}}{h^{1/3}\Delta^{1/3}} + \frac{2^n}{\Delta} \right ) \tag{by $hm>\frac{\Delta}{2n}$ from \cref{eqn:defnhmx}}\\
   & \le O^*\left (\frac{2^{2n/3}}{\Delta^{1/3}} \right ). \tag{by $h>1$ and the assumption that $\Delta\ge 2^{n/2}$}
\end{align*}
This finishes the proof of  \cref{thm:larged}.

\section{A polynomial-space algorithm}
\label{sec:poly}
We now consider $\poly(n)$-space algorithms for \pigeon. The straightforward binary search approach (described at the end of \cref{sec:prelim}) can be adapted to run in $O^*(2^n)$ time and $\poly(n)$ space: instead of using meet-in-middle (\cref{lem:mim}, which requires large space), we count the number of valid subsets $S\subseteq [n]$ by brute force in $O^*(2^n)$ time and only $\poly(n)$ space.

We improve this $O^*(2^n)$ running time using the ideas from earlier sections. Again, consider two cases depending on whether parameter $d$ from \cref{eqn:defndel} is small or large.

\begin{lemma}
   \label{lem:smalldeltalowspace}
   Given parameter $\Delta \le  2^n/(3n^2)$,  \pigeon with $d\le \Delta$ can be solved deterministically 
  in $\poly(n)$ space and $O^*(\Delta)$ time.
\end{lemma}
\begin{proof}[Proof Sketch]
  The proof is almost the same as  \cref{thm:smalldelta} (see \cref{sec:smalld}), with the only difference in Case 3 from the proof of \cref{lem:query}: instead of using meet-in-middle, here we count the valid subsets $A'\subseteq A$ by brute force in $O^*(2^{|A|}) = O^*(2^{i^*}) = O^*(\Delta)$ time and only $\poly(n)$ space.
\end{proof}

To solve the large $d$ case, we need the low-space element distinctness algorithm by Beame, Clifford, and Machmouchi \cite{BCM13} (generalized in \cite{BansalGN018}, and with a non-standard assumption removed by \cite{ChenJWW22,xinlyu}).  
This algorithm was also previously used for Subset Sum \cite{BansalGN018} and Equal Sums \cite{esaequal19}.
The following statement can be inferred from \cite[Section 4.2 (proof of Theorem 1.1)]{ChenJWW22}.
\begin{theorem}[Low-space Element Distinctness, \cite{BCM13,BansalGN018,ChenJWW22}]
   \label{theorem:bcm}
  Given random access to an integer list $a_1,\dots,a_N$ (where $a_i\in [\poly(N)]$) that contains at least one pair $(i,j)\in [N]\times [N]$ with $a_i=a_j,i\neq j$, there is a randomized algorithm that reports such a pair using $\polylog N $ working memory and 
  \[  O\left (\frac{N\sqrt{F_2}}{F_2 - N} \cdot \polylog N \right )\]
  time, where $F_2 = \sum_{i=1}^N\sum_{j=1}^N \mathbf{1}[a_i=a_j] \in [N+2,N^2]$.\footnote{We have $F_2\ge N+2$ due to the following $(N+2)$ pairs: $(1,1),(2,2),\dots,(N,N)$ and $(i,j),(j,i)$, where $a_i=a_j$ ($i\neq j$).}
\end{theorem}

\begin{lemma}
   \label{lem:largedlowspace}
Given parameter $1\le \Delta\le 2^n$, \pigeon with $d\ge \Delta$ can be solved in $O^*(2^{1.5n}/\Delta)$ time and $\poly(n)$ space by a randomized algorithm.
\end{lemma}
\begin{proof}
  Apply \cref{theorem:bcm}  to the  list $\{w(A)\}_{A\subseteq [n]}$ of length $N=2^n$ and we obtain a pair of distinct $A,A'\subseteq[n]$ with $w(A)=w(A')$ as desired. The space complexity is $\polylog(2^n) = \poly(n)$. To analyze the time complexity, note that 
  \[F_2-2^n =
  \sum_{A\subseteq [n]}\sum_{{\substack{B\subseteq [n]\\ B\neq A}}}\mathbf{1}[w(A)=w(B)] =
  \sum_{0\le t<2^n}f_t(f_t-1) \ge  \sum_{0\le t<2^n}\max\{0,f_t-1\} \stackrel{\text{Eq.~\eqref{eqn:defndel}}}{=}  d\ge \Delta,\]
  so the time bound is (ignoring $\poly(n)$ factors)
  \[  \frac{2^n\sqrt{F_2}}{F_2 - 2^n} < \frac{2^{0.5n}F_2}{F_2-2^n}  = 2^{0.5n}\left (1+ \frac{2^n}{F_2-2^n}\right )\le  2^{0.5n}\left (1+ \frac{2^n}{\Delta}\right ) \le  \frac{2\cdot 2^{1.5n}}{\Delta}\] 
  as claimed.
\end{proof}

Combining the two lemmas gives the desired result.
\begin{proof}[Proof of \cref{thm:mainlowspace}]
Set $\Delta = 2^{0.75n}$  so that the two time bounds in \cref{lem:smalldeltalowspace} and \cref{lem:largedlowspace} are balanced to $O^*(2^{0.75n})$.
Given an instance of \pigeon (with unknown $d$), we run both algorithms in parallel, and return the answer of whichever terminates first.
\end{proof}

\section{Open problems}
Allcock et al.\ \cite{esa22} proposed a modular variant of the \pigeon problem: given integers $w_1,\dots,w_n$ and a modulus $m\le 2^n-1$, find two distinct subsets $A,B\subseteq [n]$ such that $\sum_{i\in A}w_i\equiv \sum_{i\in B}w_i\pmod{m}$.  They obtained a $O^*(2^{n/2})$-time algorithm for this problem. Can this result be improved as well?  

Can we obtain faster algorithms for other problems in PPP (e.g., \cite{BanJPPR19,SotirakiZZ18})?

 \bibliography{main}
\end{document}